\documentclass[letterpaper]{scrartcl} 

\usepackage[]{group} 

\usepackage{eucal}
\usepackage{nicefrac}
\usepackage{subfigure}
\usepackage{tabularx, booktabs}
\usepackage{tikz}

\usepackage{arydshln}

\newcommand\Tstrut{\rule{0pt}{2.4ex}} 
\newcommand\Bstrut{\rule[-1.2ex]{0pt}{0pt}} 

\newcommand\bovermat[2]{%
  \makebox[0pt][l]{$\smash{\overbrace{\phantom{%
    \begin{matrix}#2\end{matrix}}}^{#1}}$}#2}

\makeatletter
\renewcommand*\env@matrix[1][*\c@MaxMatrixCols c]{%
  \hskip -\arraycolsep
  \let\@ifnextchar\new@ifnextchar
  \array{#1}}
\makeatother

\title{The Distribution of Optimal Strategies in Symmetric Zero-sum Games}

\author{Florian Brandl\\
Technical University of Munich\\
\texttt{\small brandlfl@in.tum.de}
}

\date{}

\begin{document}

\maketitle

\begin{abstract}
	 Given a skew-symmetric matrix, the corresponding two-player symmetric zero-sum game is defined as follows: one player, the row player, chooses a row and the other player, the column player, chooses a column. The payoff of the row player is given by the corresponding matrix entry, the column player receives the negative of the row player. A randomized strategy is optimal if it guarantees an expected payoff of at least 0 for a player independently of the strategy of the other player.
	 We determine the probability that an optimal strategy randomizes over a given set of actions when the game is drawn from a distribution that satisfies certain regularity conditions. The regularity conditions are quite general and apply to a wide range of natural distributions.
\end{abstract}

 \noindent\textbf{JEL Classifications Code: }C62, C72

\section{Introduction}

A (two-player) zero-sum game is played on a matrix where the row player chooses a row and the column player chooses a column. The payoff of the row player is given by the corresponding matrix entry, the column player receives the negative of the row player. Both players may randomize over their actions. Von Neumann's minimax theorem shows that every zero-sum game admits a value, i.e., the row player can guarantee an expected payoff for himself that is equal to the negative of the expected payoff that the column player can guarantee for himself. A strategy that maximizes the minimal guaranteed expected payoff of a player is a maximin strategy for this player. Pairs of maximin strategies correspond to Nash equilibria of the game. We will refer to maximin strategies as \emph{optimal strategies}.

A zero-sum game is \emph{symmetric} if the corresponding payoff matrix is skew-symmetric. Thus, both players have the same set of actions and every maximin strategy of the row player is also a maximin strategy of the column player and \emph{vice versa}. Moreover, both players can achieve a payoff of at least $0$ by playing the same strategy as the other player. This also implies that the value of a symmetric zero-sum game is $0$. Symmetric zero-sum games can be associated with weighted digraphs where the vertices correspond to actions and the weights of the edges are the payoffs from choosing the corresponding actions.

Symmetric zero-sum games appear in many areas of natural science such as biology, physics, and chemistry. We only give two examples here. In evolutionary biology, they can be used to model population dynamics among multiple species with actions corresponding to species and payoffs corresponding to the probabilities that an individual from one species ``beats'' an individual from another species; the probabilities in an optimal strategy specify the fractions of individuals from each species in a stable state. Hence, the support of an optimal strategy constitutes the set of species that survive in a stable state. In quantum physics, symmetric zero-sum games appear in bosonic systems where different quantum states take the role of actions and the transition probabilities from one state to another form the payoffs. \citet{KWKF15a} consider the support of optimal strategies in these games to determine which states become condensates. 

In this paper we show that, for every set of actions $S$, the probability that a symmetric zero-sum game admits an optimal strategy with support $S$ is $2^{-(n-1)}$ if $S$ has odd cardinality and $0$ otherwise, where $n$ is the total number of actions. In particular, this probability only depends on the parity of $S$. 
This stems from the fact that a skew-symmetric matrix of odd size cannot have full rank.\footnote{\label{note:1}A skew-symmetric matrix of odd size $G$ cannot have full rank, since $\det(G) = \det(G^T) = \det(-G) = (-1)^n\det(G) = -\det(G)$ and, hence, $\det(G) = 0$.}
For the proof of this result we assume that the distribution of games is symmetric and regular.
A distribution is \emph{symmetric} if it is invariant under negation of all payoffs when the row player chooses an action from a certain set and the column player chooses an action from the complement set. Intuitively, this condition prescribes that, in the graph representation of the game, reversing all edges between $S$ and its complement set does not change the probability of the game being chosen.
A distribution is \emph{regular} if a randomly chosen game almost surely admits a unique optimal strategy. 
We assume throughout that games are drawn from a symmetric, regular distribution. 

Related questions have been studied for various classes of games. \citet{Wils71a} showed that the number of Nash equilibria is finite and odd for almost all $n$-person normal form games. A different proof of the same statement was given by \citet{Hars73a}. \citet{McLe05a} derived a formula for the expected number of Nash equilibria in which players play a certain set of actions with positive probability in normal form games. His model assumes that the payoffs of all players are independent and distributed uniformly over the unit sphere. If games are distributed such that Nash equilibria are almost surely unique, the expected number of Nash equilibria with given support is equal to the probability that the game admits a Nash equilibrium with this support. Thus, our result can also be phrased as determining the expected number of Nash equilibria with given support.
Follow-up work by \citet{McBe05a} has derived a formula for the expected number of Nash equilibria of a random two-player normal form game. Similarly to \citet{McLe05a}, they assume that the payoffs of both players are drawn independently from a uniform distribution on the unit sphere.  

In zero-sum games every convex combination of Nash equilibria is again a Nash equilibrium. Hence there is either a unique Nash equilibrium or infinitely many. However, \citeauthor{Wils71a}'s theorem does not imply that Nash equilibria are almost surely unique in low dimensional subclasses of normal form games, e.g., zero-sum games, symmetric zero-sum games, or tournament games.\footnote{Tournament games are symmetric zero-sum games in which all off-diagonal payoffs are either $1$ or $-1$.} \citet{FiRy92a} showed that every tournament game admits a unique optimal strategy and, hence, a unique Nash equilibrium. This result was generalized by \citet{LLL97a} to symmetric zero-sum games where all payoffs are odd integers, and by \citet{LeBr05a} to symmetric zero-sum games where all payoffs satisfy a more general congruency condition. 
Closest to the present paper is the work of \citet{Robe04a}, who proves the same formula that is derived in this paper for random symmetric zero-sum games and a somewhat less general class of distributions. He assumes that the payoffs are drawn i.i.d. from a distribution that is symmetric about $0$, i.e., a distribution with even density function. Rather than making assumptions about the distributions of single payoffs, we only impose assumptions about the distribution of the entire game matrix. In particular, we do not require payoffs to be independent or identically distributed. 

For (not necessarily symmetric) zero-sum games, the situation is less clear. Experiments by \citet{FaMa87a} suggest that the size of the optimal strategy of a zero-sum game chosen uniformly at random approximately follows a binomial distribution that chooses half of the actions in expectation. \citet{Jona04a} showed that optimal strategies are almost surely unique if the payoffs are drawn from continuous i.i.d. random variables that are symmetric about $0$. Moreover, he proved that the expected fraction of actions in the support of an optimal strategy is close to $\nicefrac{1}{2}$ when the number of actions goes to infinity. \citet{Robe06a} considers zero-sum games where payoffs follow independent and identical Cauchy distributions. Remarkably, he derives a closed form formula for the probability that the pair of optimal strategies of a random game has a given support.  

The proof of our main result (\Cref{cor:main}) is divided into three statements. In \Cref{thm:totallymixed} we determine the probability that an optimal strategy puts positive probability on all actions, i.e., the probability that a game admits a totally mixed optimal strategy. \Cref{lem:subgame} establishes that the distribution of the subgame where both players are restricted to the same set of actions is symmetric and regular if the distribution of the full game is symmetric and regular. As a consequence of these two statements we get the probability that such a subgame admits a totally mixed optimal strategy. Lastly, in \Cref{thm:extend} we determine the probability that an optimal strategy of a subgame is optimal in the full game. The probability that a game admits an optimal strategy with given support may then be derived easily. In this sense the structure of the proof is very similar to \citeauthor{McLe05a}'s \citeyearpar{McLe05a} argument.

Finally, we will argue that symmetric, regular distributions occur naturally. For example, if the payoffs of the game follow independent normal distributions the distribution of games is symmetric and regular. More generally, we will show that every absolutely continuous distribution is regular (\Cref{thm:almostsurely}). As noted before, every tournament game admits a unique optimal strategy. Thus, the uniform distribution over all tournament games is symmetric and regular. As a consequence, \Cref{cor:main} implies a result of \citet{FiRe95a}, who determine the probability that the optimal strategy of a random tournament game uses $k$ actions with non-zero probability.

\section{Preliminaries}

A \emph{zero-sum game} $G$ is a matrix in $\mathbb{R}^{M\times N}$, where $M$ and $N$ are the sets of actions for the row and column player, respectively.
We will write $|N| = n$ for short.
The matrix entry $G_{ij}$ represents the payoff of the \emph{row} player if he chooses row $i$ and the column player chooses column $j$.
The set of all probability distributions over a finite set $S$ is denoted by $\Delta(S)$, i.e., $\Delta(S) = \{p\in\mathbb{R}^S\colon p\ge 0 \text{ and } \sum_{i \in S} p_i = 1\}$. 
A \emph{(randomized) strategy} for the row player or the column player is a probability distribution on $M$ or $N$, respectively. 
The \emph{support} $p_+$ of a strategy $p\in\Delta(N)$ is the set of actions to which $p$ assigns positive probability, i.e., $p_+ = \{i\in N\colon p_i > 0\}$. 
For vectors $v\in\mathbb{R}^N$, we additionally define $v_- = \{i\in N\colon v_i < 0\}$. 
A strategy $q^\ast$ is a \emph{maximin strategy} for the row player if it maximizes his minimum expected payoff, i.e., 
\[
\min_{p\in \Delta(N)} {q^\ast}^TGp\ge \max_{q\in\Delta(M)}\min_{p\in \Delta(N)} q^TGp\text.
\]
Maximin strategies for the column player are defined analogously.
By the minimax theorem the minimum expected payoff of the row player when he plays a maximin strategy is equal to the negative of the minimum expected payoff of the column player when he plays a maximin strategy.
This payoff is called the \emph{value} of the game.
The set of pairs $(q^\ast,p^\ast)$ such that $q^\ast$ and $p^\ast$ are maximin strategies of the row player and the column player, respectively, is the set of Nash equilibria of the game.
We say that a strategy is an \emph{optimal} strategy of a player if it is a maximin strategy.
Note that the set of optimal strategies of each player is convex, since they are the sets of solutions to linear programs.
\citet{Ragh94a} proved that every action of the row player that yields the same payoff as an optimal strategy against all optimal strategies of the column player is played with positive probability in some optimal strategy of the row player. This is known as the equalizer theorem.

	\begin{proposition}[\citealp{Ragh94a}]\label{thm:equalizer}
		Let $G$ be a game with value $v\in\mathbb{R}$ and $i\in M$. If $(Gp)_i = v$ for all optimal strategies $p$ of the column player, then there is an optimal strategy $q^\ast$ of the row player with $q^\ast_i > 0$. 
	\end{proposition}

	Following \citet{Hars73b}, an equilibrium $(q^\ast,p^\ast)$ is \emph{quasi-strict} if every action of the row player that is in the support of $q^\ast$ yields strictly more expected payoff against $p^\ast$ than every action that is not in the support of $q^\ast$ (and similarly for the column player).\footnote{\citeauthor{Hars73b} introduced the concept of quasi-strong equilibria, which however was referred to as quasi-strict equilibria in subsequent papers to avoid confusion with Aumann's notion of strong equilibria \citep[][]{Auma59a}.} It is a well-known fact that if a game only admits quasi-strict equilibria, then it in fact has a unique equilibrium. \Cref{lem:quasi-strict} shows that the converse is also true, i.e., if a game admits an equilibrium that is not quasi-strict, then it cannot be the unique equilibrium of the game.
	The proof of \Cref{lem:quasi-strict} makes use of the equalizer theorem.
	
	\begin{lemma}\label{lem:quasi-strict}
		Let $G$ be a game and $(q^\ast,p^\ast)$ an equilibrium of $G$. If $(q^\ast,p^\ast)$ is not quasi-strict, then $G$ has multiple equilibria. 
	\end{lemma}
	
	\begin{proof}
		Let $v = (q^\ast)^TGp^\ast$ be the value of $G$. 
		Assume for contradiction that $(q^\ast,p^\ast)$ is the unique equilibrium of $G$ and $(q^\ast,p^\ast)$ is not quasi-strict, i.e., without loss of generality there is $i\not\in q^\ast_+$ such that $(Gp^\ast)_i = v$.
		Then it follows from \Cref{thm:equalizer} that the row player has an optimal strategy $\hat q^\ast$ with $\hat q^\ast_i > 0$. In particular, $\hat q^\ast\neq q^\ast$. Then $(\hat q^\ast,p^\ast)$ is another equilibrium of $G$. 
		This contradicts uniqueness of $(q^\ast,p^\ast)$.
	\end{proof}

\section{Symmetric Zero-sum Games}

Symmetric zero-sum games constitute a class of zero-sum games with particularly nice properties.
A zero-sum game is \emph{symmetric} if $G$ is skew-symmetric, i.e., $G = -G^T$. 
For brevity, we will simply use \emph{game} to refer to a symmetric zero-sum game for the remainder of the paper. 
The set of all games is denoted by $\mathcal{G}$. 
Symmetry implies that the sets of optimal strategies of both players coincide.
We will hence simply use the term optimal strategy without referring to a specific player.

For a set of actions $S\subseteq N$, a game $G\in\mathcal{G}$, and a vector $v\in\mathbb{R}^N$, we denote by $G_S = (G_{ij})_{i,j\in S}$ and $v_S = (v_i)_{i\in S}$ the sub-matrix and sub-vector induced by $S$, respectively.
To simplify the proofs, we introduce special notation for particular classes of games. 
The set of games where $G_S$ has multiple optimal strategies is denoted by $\mathcal{G}_S^{>1}$, i.e.,
\[
\mathcal{G}_S^{>1} = \{G\in\mathcal{G}\colon G_S \text{ has two distinct optimal strategies}\}
\]
Note that $\mathcal{G}_N^{>1}$ contains all games with two distinct optimal strategies. We write $\mathcal{G}^{>1}$ short for $\mathcal{G}_N^{>1}$.

A strategy is \emph{totally mixed} if all actions are played with strictly positive probability. The set of all games where $G_S$ has a totally mixed optimal strategy is denoted by $\mathcal{G}_S$, i.e.,
\[
\mathcal{G}_S = \{G\in\mathcal{G}\colon G_S \text{ has an optimal strategy $p$ with } p_+ = S\}
\]
Lastly, we define the set of all games that admit an optimal strategy with support $S$, i.e.,
\[
\mathcal{G}_S^\ast = \{G\in\mathcal{G}\colon G \text{ has an optimal strategy $p$ with } p_+ = S\}\text.
\]
Since every optimal strategy of the full game is also an optimal strategy of the subgame induced by its support, $\mathcal{G}_S^\ast$ is a subset of $\mathcal{G}_S$.

We assume that games are drawn from a probability distribution $\mathcal{X}$.
By $X$ we denote a random variable with distribution $\mathcal{X}$, i.e., $X\sim\mathcal{X}$.
For a set of games $\mathcal{G}'\subseteq\mathcal{G}$, let $P_X(\mathcal{G}')$ be the probability that a realization of $X$ is in $\mathcal{G}'$.
To establish our results, we require that $\mathcal{X}$ satisfies two regularity conditions.  
For $S\subseteq N$, we define the automorphism $\Phi_S$ on $\mathcal{G}$ such that, for all $i,j\in N$,
\[
	(\Phi_S(G))_{ij} = 
	\begin{cases}
		G_{ij} \quad&\text{if } i,j\in S \text{ or } i,j\in N\setminus S\text{, and}\\
		-G_{ij}\quad&\text{otherwise.}
	\end{cases}
\]
Then $\mathcal{X}$ is \emph{symmetric} if it is invariant under $\Phi_S$ for every $S\subseteq N$, i.e., $P_X(\mathcal{G}') = P_X(\Phi_S(\mathcal{G}'))$ for every $\mathcal{G}'\subseteq\mathcal{G}$.
Observe that, for all $S,T\subseteq N$, we have $\Phi_S\circ\Phi_T = \Phi_{S\mathrel{\Delta} T}$, where $\Delta$ is the symmetric difference of $S$ and $T$.
Furthermore, $\Phi_S = \Phi_{N\setminus S}$ for all $S\subseteq N$.
As a consequence, $(\{\Phi_S\colon S\subseteq N\},\circ)$ is a group with neutral element $\Phi_\emptyset$ such that every element is self-inverse.
The fact that $S\Delta T = T\Delta S$ implies that this group is abelian.
Moreover, we require $\mathcal{X}$ to be \emph{regular} in the sense that $X$ almost surely admits a unique optimal strategy or, formally, $P_X(\mathcal{G}^{>1}) = 0$.

\section{The Result}

The main result is obtained in \Cref{cor:main} and states the following: if games are drawn from a symmetric, regular probability distribution then, for every set of actions $S$, the probability that a symmetric zero-sum game admits an optimal strategy with support $S$ is $2^{-(n-1)}$ if $S$ has odd cardinality and $0$ if $S$ has even cardinality.
\Cref{cor:main} is an obvious consequence of \Cref{lem:subgame} and Propositions~\ref{thm:totallymixed} and~\ref{thm:extend}.
		The first lemma shows that every strategy that is the unique optimal of some game puts positive probability on an odd number of actions. This does not hold for non-symmetric zero-sum games. E.g., the game known as matching pennies has a unique optimal strategy of size $2$.
	
	\begin{lemma}\label{lem:odd}
		Let $G$ be a game and $p$ be the unique optimal strategy of $G$. Then the support of $p$ has odd cardinality.
	\end{lemma}
	
	\begin{proof}
		Assume for contradiction that $p_+$ has even cardinality.
		Let $p_+ = S$.
		Since $p$ is the unique optimal strategy of $G$, it follows from \Cref{lem:quasi-strict} that $(Gp)_i < 0$ for all $i\not\in S$. 
		By definition of $S$, $|S\setminus\{i\}|$ is odd for $i\in S$.
		Hence, $G_{S\setminus\{i\}}$ does not have full rank, i.e., there is $v\in\mathbb{R}^n\setminus\{0\}$ with $v_-\cup v_+\subseteq S\setminus\{i\}$ and $G_{S\setminus\{i\}}v_{S\setminus\{i\}} = 0$.
		Assume without loss of generality that $(Gv)_i\le 0$ (otherwise we take $-v$).
		Then, for $\epsilon > 0$ small enough, we have that $p^\epsilon = (1-\epsilon)p+\epsilon v\ge 0$ and $Gp^\epsilon\le 0$, i.e., $\nicefrac{p^\epsilon}{|p^\epsilon|}$ is an optimal strategy of $G$.
		This contradicts uniqueness of $p$.
	\end{proof}

	Now we prove an equation that will be useful for the upcoming proofs. 
	
	\begin{lemma}\label{lem:equation}
		Let $G\in\mathcal{G}$, $v\in\mathbb{R}^n$, and $S\subseteq N$. Then,
		\[
			\Phi_S(G)\Phi_S(v) = \Phi_S(Gv)\text.
		\]
	\end{lemma}
	
\begin{proof}
	This is readily checked by verifying the following sequence of equalities:
		\vspace{1.4ex}
	\[
	\Phi_S(G)\Phi_S(v) =
	 \begin{pmatrix}[c:c]
	 	\bovermat{S}{G_{ij}} & \bovermat{N\setminus S}{-G_{ij}}\Bstrut\\
		\hdashline   
		-G_{ij} & G_{ij}\Tstrut
	\end{pmatrix}
	\cdot
	\begin{pmatrix}[c]
	 	-v_j\Bstrut\\
		\hdashline   
		v_j\Tstrut
	\end{pmatrix}
		 =
	\begin{pmatrix}[c]
		-(Gv)_i\Bstrut\\
		\hdashline   
		(Gv)_i\Tstrut
	 \end{pmatrix}
	 =
	\Phi_S(Gv)\text.
	\]
\end{proof}

	For regular distributions, it follows quickly from \Cref{lem:odd} that the probability that a game has an optimal strategy with even support size is $0$. If the distribution is also symmetric, it turns out that the probability that a game has an optimal strategy with given support of odd size is independent of the chosen support. This is again specific to symmetric zero-sum games and does not hold in general for zero-sum games.

\begin{proposition}\label{thm:totallymixed}
	Let $\mathcal{X}$ be symmetric and regular. Then the probability that $X$ has a totally mixed optimal strategy is
	\[
		\begin{cases}
			0\quad&\text{if $n$ is even, and}\\
			2^{-(n-1)}\quad&\text{if $n$ is odd.}
		\end{cases}
	\]
\end{proposition}
	
\begin{proof}
	First we consider the case when $n$ is even. Assume that a game $G\in\mathcal{G}$ has a totally mixed optimal strategy. It follows from \Cref{lem:odd} that $G$ has multiple optimal strategies. Thus, $\mathcal{G}_N\subseteq\mathcal{G}^{>1}$, which implies that $P_X(\mathcal{G}_N) = 0$.
	
	Now assume that $n$ is odd.
	For all $S\subseteq N$, let $\mathcal{G}_S^=$ be the set of games with a vector $v$ in the null space such that $v_+ = N\setminus S$.
	Note that $\mathcal{G}^=_\emptyset$ is the set of all games with a totally mixed optimal strategy.
	Every game is in $\mathcal{G}_S^=$ for some $S\subseteq N$, since a skew-symmetric matrix of odd size cannot have full rank (cf. Footnote~\ref{note:1}).
	For $S\subseteq N$, let $\mathcal{G}_S^0\subseteq\mathcal{G}_S^=$ be the set of games with a vector $v$ in the null space such that $v_+ = N\setminus S$ and $v_i = 0$ for some $i\in N$.
	It follows from \Cref{lem:equation} that $\Phi_S(G)\Phi_S(v) = 0$.
	Since $\Phi_S(v)\ge 0$ and $\Phi_S(v)_i = 0$, it follows from \Cref{lem:quasi-strict} that $\Phi_S(G)$ has multiple optimal strategies.
	Thus, $\Phi_S(\mathcal{G}_S^0)\subseteq\mathcal{G}^{>1}$.
	By symmetry of $\mathcal{X}$, we then have $P_X(\mathcal{G}_S^0) = P_X(\Phi_{S}(\mathcal{G}_S^0)) \le P_X(\mathcal{G}^{>1}) = 0$.
	Hence, vectors in the null space almost surely have no entries equal to $0$.
	This implies that $P_X(\mathcal{G}^=_S) = P_X(\Phi_{S\mathrel{\Delta}T}(\mathcal{G}^=_S)) = P_X(\mathcal{G}^=_T)$ for all $S,T\subseteq N$.
	Moreover, $\mathcal{G}^=_S$ and $\mathcal{G}^=_{N\setminus S}$ only differ by a null set, since $v_+ = N\setminus (-v)_+$ if $v$ has no zero entries.
	Hence $P_X(\mathcal{G}^=_S) = P_X(\mathcal{G}^=_{S}\cap \mathcal{G}^=_{N\setminus S})$ for all $S\subseteq N$.
	Now we show that $X$ almost surely has rank $n-1$.
	From before we know that $X$ has rank at most $n-1$.
	If $X$ has rank less than $n-1$, there are distinct $v,w\in\mathbb{R}^n$ such that $Xv = Xw = 0$.
	But then $\lambda v + (1-\lambda) w$ is in the null space of $X$ and has an entry equal to $0$ for some $\lambda\in\mathbb{R}$.
	This is a probability zero event as shown above.
	Hence, $X$ almost surely has rank $n-1$.
	This implies that $P_X(\mathcal{G}^=_S\cap\mathcal{G}^=_T) = 0$ for all $S,T\subseteq N$ with $S\neq T$ and $S\neq N\setminus T$.
	Together, we get $P_X(\mathcal{G}^=_S) = 2^{-(n-1)}$ for all $S\subseteq N$.
\end{proof}

It was already observed by \citet{Kapl45a} that a game of even size cannot have a unique, totally mixed optimal strategy, which follows from the fact that the rank of a skew-symmetric matrix is even.\footnote{The rank of a skew-symmetric matrix is even, since skew-symmetric matrices of odd size cannot have full rank (cf. Footnote~\ref{note:1}).} Moreover, \citet{Kapl95a} shows that a game admits a unique, totally mixed optimal strategy if and only if the principal Pfaffians of the corresponding matrix alternate in sign.\footnote{The $i$th principal Pfaffian is the Pfaffian of the matrix obtained by deleting the $i$th row and $i$th.} This result allows for a more algebraic but arguably less instructive proof of \Cref{thm:totallymixed}.

\begin{lemma}\label{lem:subgame}
	Let $S\subseteq N$. If $\mathcal{X}$ is symmetric and regular, then $\mathcal{X}_S$ is a symmetric and regular.
\end{lemma}

\begin{proof}
	Let $S\subseteq N$ and $\mathcal{X}$ be symmetric and regular.
	First we show that $\mathcal{X}_S$ is symmetric.
	To this end, let $T\subseteq S$ and $\mathcal{G}_S'\subseteq\mathcal{G}_S$.
	Then,
	\begin{align*}
		P_{X_S}(\mathcal{G}_S') &= P_{X}(\{G\in\mathcal{G}\colon G_S\in\mathcal{G}_S'\}) = P_{X}(\Phi_T(\{G\in\mathcal{G}\colon G_S\in\mathcal{G}_S'\}))\\
		&= P_{X}(\{\Phi_T(G)\colon G\in\mathcal{G} \text{ and } G_S\in\mathcal{G}_S'\}) = P_{X}(\{G\in\mathcal{G}\colon \phi_T(G_S)\in\mathcal{G}_S'\})\\
		&= P_{X}(\{G\in\mathcal{G}\colon G_S\in\phi_T(\mathcal{G}_S')\}) = P_{X_S}(\phi_T(\mathcal{G}_S'))\text.
	\end{align*}
	The first and the last equality follow from the definition of $X_S$.
	The second equality holds by symmetry of $\mathcal{X}$.
	The third equality uses the definition of $\Phi_T$ as applied to sets of games.
	The forth equality holds since $\Phi_T$ is self-inverse and since $\Phi_T$ commutes with restriction to $S$.
	Lastly, the fifth equality again holds since $\Phi_T$ is self-inverse.
	
	Now we show by induction over $|S|$ that $\mathcal{X}_S$ is regular.
	If $S = N$ the statement is clear by the hypothesis of the lemma.
	For the induction step, let $S\subsetneq N$ and assume that $\mathcal{X}_T$ is regular for all $T\subseteq N$ with $|T| > |S|$.
	Assume for contradiction that $\mathcal{X}_S$ is not regular, i.e., $P_{X_S}(\mathcal{G}_S^{>1}) > 0$.
	Let $i\in N\setminus S$ and $S^{i} = S\cup\{i\}$.
	Then, we have that $P_{X_{S^i}}(\{G_{S^i}\in\mathcal{G}_{S^i}\colon G_S\in\mathcal{G}_S^{>1}\}) = P_{X_S}(\mathcal{G}_S^{>1}) > 0$.
	We define
	\[
		\mathcal{G}_{S^i}^{-} = \{G\in\mathcal{G}_{S^i}\colon p\in\Delta(S^i) \text{ with } p_+ = S, G_Sp_S\le 0 \text{ and, } (Gp)_i\le 0 \text{ for some } p\in\Delta(S^i)\}\text,
	\]
	with $\mathcal{G}_{S^i}^{+}$ defined by replacing the last $\le$ by $\ge$.
	Since $\mathcal{X}_{T}$ is symmetric for every $T\subseteq N$, it follows that $P_{X_{S^i}}(\mathcal{G}_{S^i}^{-}) = P_{X_{S^i}}(\mathcal{G}_{S^i}^{+})$.
	Moreover, $\mathcal{G}_{S^i}^{-}\cup\mathcal{G}_{S^i}^{+} = \{G_{S^i}\in\mathcal{G}_{S^i}\colon G_S\in\mathcal{G}_S^{>1}\}$ and, hence, $P_{X_{S^i}}(\mathcal{G}_{S^i}^{-})> 0$.
	Now let $G\in\mathcal{G}_{S^i}^{-}$.
	If there is $p\in\Delta(S^i)$ such that $G_Sp_S\le 0$ and $(Gp)_i = 0$, then if follows from \Cref{lem:quasi-strict} that $G$ has multiple optimal strategies.
	If $(Gp)_i < 0$, there is $q\in\Delta(S_i)$ such that $q_S\neq p_S$ and $G_Sq_S\le 0$.
	Such a $q$ exists since $G_S\in\mathcal{G}_S^{>1}$ by definition.
	But then $(1-\lambda)p_{S^i} + \lambda q_{S^i}$ is another optimal strategy of $G$ for small $\lambda > 0$.
	In any case, $G$ has two distinct optimal strategies.
	Thus, we have
	\[
		P_{X_{S^i}}(\mathcal{G}_{S^i}^{>1})\ge P_{X_{S^i}}(\mathcal{G}_{S^i}^{-}) > 0\text,
	\]
	which contradicts the induction hypothesis that $\mathcal{X}_{S^i}$ is regular.
\end{proof}

By combining the last two statements we get the probability that $X_S$ admits a totally mixed optimal strategy. In the next proposition we determine the probability that $X$ has an optimal strategy with support $S$ given that $X_S$ has a totally mixed optimal strategy.

\begin{proposition}\label{thm:extend}
	Let $\mathcal{X}$ be symmetric and regular and $S\subseteq N$. Then $P_X(\mathcal{G}_S^\ast\,|\, \mathcal{G}_S) = 2^{-(n-|S|)}$.
\end{proposition}

\begin{proof}
	Let $\mathcal{X}$ be symmetric and regular and $S\subseteq N$.
	Recall that $\mathcal{G}_S$ is the set of all games where $G_S$ has a totally mixed optimal strategy.
	Moreover, we define $\mathcal{G}_S(T)$ to be the set of all games where $G_S$ has a totally mixed optimal strategy such that the set of actions yielding positive payoff corresponds exactly to the rows in $T$, i.e.,
	\begin{align*}
		\mathcal{G}_S(T) = \{&G\in\mathcal{G}\colon p_S \text{ is a totally mixed optimal strategy of } G_S \text{ and } (Gp)_+ = T \\
		&\text{for some } p\in\Delta(N)\}\text.
	\end{align*}
	Note that $\mathcal{G}_S(\emptyset) = \mathcal{G}_S^\ast$ and $\mathcal{G}_S(T)$ is non-empty only if $T\subseteq N\setminus S$.
	It follows from \Cref{lem:equation} that $\Phi_T(\mathcal{G}_S(T))\subseteq\mathcal{G}_S^*$ for all $T\subseteq N\setminus S$.
	For $G\in\mathcal{G}_S^*\setminus \Phi_T(\mathcal{G}_S(T))$ we have that $(Gp)_i = 0$ for some $i\in T$.
	Then it follows from \Cref{lem:quasi-strict} that $G$ has multiple optimal strategies.
	Thus, by symmetry of $\mathcal{X}$, we have $P_X(\mathcal{G}_S(T)) = P_X(\Phi_T(\mathcal{G}_S(T))) = P_X(\mathcal{G}_S^\ast)$ for all $T\subseteq N\setminus S$.
	Since, by \Cref{lem:subgame}, $X_S$ almost surely has a unique optimal strategy, we also have that $P_X(\mathcal{G}_S(T)\cap \mathcal{G}_S(T')) = 0$ for all distinct $T,T'\subseteq N\setminus S$.
	Since $N\setminus S$ has $2^{n-|S|}$ distinct subsets, we get $P_X(\mathcal{G}_S^\ast\,|\, \mathcal{G}_S) = 2^{-(n-|S|)}$.
\end{proof}

The main result easily follows from \Cref{lem:subgame} and \Cref{thm:totallymixed,thm:extend}.

\begin{theorem}\label{cor:main}
	Let $\mathcal{X}$ be symmetric and regular. Then, for every $S\subseteq N$, the probability that $X$ has an optimal strategy with support $S$ is
	\[
	\begin{cases}
		0 \quad&\text{if } |S| \text{ is even, and}\\
		2^{-(n-1)} \quad &\text{if } |S| \text{ is odd.}
	\end{cases}
	\]
\end{theorem}
Observe that $N$ has $2^{n-1}$ subsets of odd size. Hence, the probabilities above sum to $1$.

We show that the assumptions in \Cref{cor:main} are indeed necessary.
It is easily seen that the assumptions that $\mathcal X$ is symmetric and regular are necessary.
The conclusion of \Cref{cor:main} does not hold for the distribution that returns the game with all payoffs equal to $0$ with probability one, even though it is symmetric.
Neither does \Cref{cor:main} hold for the distribution that returns the classic game of ``rock, paper, scissors'' with probability one, even though it is regular.
\Cref{cor:main} may also fail if (not necessarily symmetric) zero-sum games are sampled from a symmetric and regular distribution. If $|M| = |N| = 2$ and each entry in the payoff matrix is drawn from a standard normal distribution, the probability that an optimal strategy of the row player in the resulting zero-sum game has full support is one third.

	Lastly, we show that a distribution is regular if it is absolutely continuous (w.r.t. the Lebesgue measure). In particular, a distribution is absolutely continuous if all payoffs are independent, absolutely continuous random variables. This implies that, if payoffs are drawn from independent, absolutely continuous distributions that are symmetric about $0$, e.g., normal distributions or uniform distributions on intervals that are symmetric about $0$, then the induced distribution on the set of games is regular and symmetric.

	\begin{proposition}\label{thm:almostsurely}
		If $\mathcal{X}$ is absolutely continuous, then it is regular.
	\end{proposition}

	\begin{proof}
		Let $\mathcal{X}$ be absolutely continuous.
		Moreover, let $G\in\mathcal{G}^{>1}$.
		We will show that $G$ has a singular square sub-matrix of even size.
		As discussed before, every game that has multiple optimal strategies admits an optimal strategy that is not quasi-strict.
		Let $p$ be an optimal strategy of $G$ that is not quasi-strict, i.e., there is $i\not\in p_+$ such that $(Gp)_i = 0$.
		Let $p_+ = S$.
		Then, $G_{S}p_{S} = 0$ and $G_{S\cup\{i\}}p_{S\cup\{i\}} = 0$.
		So either $G_{S}$ or $G_{S\cup\{i\}}$ is a singular square sub-matrix of even size.

		Since $\mathcal{X}$ is absolutely continuous, every even-sized square sub-matrix of $X$ is almost surely regular.
		Thus, $X$ almost surely admits a unique optimal strategy, i.e., $\mathcal{X}$ is regular.
	\end{proof}

\section*{Acknowledgements}
This material is based upon work supported by Deutsche Forschungsgemeinschaft under grant {BR~2312/10-1} and TUM Institute for Advanced Study through a Hans Fischer Senior Fellowship.
An earlier version of this paper has been presented at the 5th World Congress of the Game Theory Society in Maastricht (July 2016).
The author would like to thank Felix Brandt, Fedor Sandomirskiy, the anonymous advisory editor, and the anonymous reviewer for helpful comments.

\end{document}